\documentclass[12pt]{article}
\usepackage[intlimits]{amsmath}
\usepackage{amssymb,amsthm,slashed,mathabx}

\usepackage[T1]{fontenc}
\usepackage[cp1250]{inputenc}

\newcommand{\<}{\langle}
\renewcommand{\>}{\rangle}
\newcommand{\Sc}{\mathcal{S}}
\newcommand{\F}{\mathcal{F}}
\newcommand{\B}{\mathcal{B}}
\newcommand{\C}{\mathcal{C}}
\newcommand{\Hc}{\mathcal{H}}
\newcommand{\M}{\mathcal{M}}
\newcommand{\hM}{\widehat{M}}
\newcommand{\mR}{\mathbb{R}}

\newcommand{\vp}{\vec{p}}
\newcommand{\vx}{\vec{x}}
\newcommand{\vy}{\vec{y}}

\newcommand{\al}{\alpha}
\newcommand{\ga}{\gamma}
\newcommand{\ep}{\epsilon}
\newcommand{\la}{\lambda}

\newcommand{\w}{\omega}
\newcommand{\W}{\Omega}
\newcommand{\vep}{\varepsilon}
\newcommand{\vph}{\varphi}
\newcommand{\dsp}{\displaystyle}

\newcommand{\ti}{\widetilde}
\newcommand{\ov}{\overline}
\newcommand{\wh}{\widehat}
\newcommand{\wch}{\widecheck}
\newcommand{\p}{\partial}
\newcommand{\con}{\mathrm{const}}

\DeclareMathOperator{\supp}{supp}
\newtheorem{thm}{Theorem}
\newtheorem{pr}[thm]{Proposition}
\newtheorem{lem}[thm]{Lemma}
\newtheorem{col}[thm]{Corollary}
\newtheorem{dfn}{Definition}
\newtheorem{rem}[thm]{Remark}

\makeatletter
\renewcommand\@biblabel[1]{#1.}
\makeatother

\title{On energy-momentum transfer of quantum fields}
\author{Andrzej Herdegen\thanks{e-mail: herdegen@th.if.uj.edu.pl}\\
{\it Institute of Physics, Jagiellonian University,}\\
{\it Reymonta 4, 30-059 Kraków, Poland}}
\date{}

\begin{document}

\maketitle

\begin{abstract}
We prove the following theorem on bounded operators in quantum field theory: if $\|[B,B^*(x)]\|\leq \con D(x)$, then $\|B^k_\pm(\nu)G(P^0)\|^2\leq\con\int D(x-y)d|\nu|(x)d|\nu|(y)$, where $D(x)$ is a~function weakly decaying in spacelike directions, $B^k_\pm$ are creation/annihilation parts of an appropriate time derivative of $B$, $G$ is any positive, bounded, non-increasing function in $L^2(\mR)$, and $\nu$ is any finite complex Borel measure; creation/annihilation operators may be also replaced by $B^k_t$ with $\wch{B^k_t}(p)=|p|^k\wch{B}(p)$. We also use the notion of energy-momentum scaling degree of $B$ with respect to a~submanifold (Steinmann-type, but in momentum space, and applied to the norm of an operator). These two tools are applied to the analysis of singularities of $\wch{B}(p)G(P^0)$. We prove, among others, the following statement (modulo some more specific assumptions): outside $p=0$ the only allowed contributions to this functional which are concentrated on  a~submanifold (including the trivial one -- a~single point) are Dirac measures on hypersurfaces (if the decay of $D$ is not to slow).

\vspace{1ex}
\noindent
keywords: quantum field theory, translation automorphism group, spectral properties

\vspace{1ex}
\noindent
MSC2010: 46L40, 81T05

\end{abstract}

\section{Introduction}\label{int}

One of the most fundamental features of the relativistic quantum field theory in the Minkowski spacetime is the transformation of observables, or more generally `quantum fields', by an  automorphism group representing spacetime translations. This action is always assumed to be implemented by the action of a~continuous representation of translations by unitaries in a~Hilbert space. While the spectral properties of unitary representations are under very good control, especially in case of local observables and vacuum representation (see~\cite{bo96}, \cite{ha92}), the problem of spectral properties of automorphisms themselves is far less worked out (but see a~recent analysis~\cite{dy10}).

In this article we give some contributions to this topic. The main tool is Theorem~\ref{GB}, which uses a~more technical Theorem~\ref{Bk}, both proved in Section~\ref{bound}. It states that under very weak assumptions on the decay of commutators in spacelike distances a~split into creation/annihilation parts of an appropriate (generalized) time derivative of an operator is possible. It also gives bounds on norms of these operators, when composed with a~weakly decreasing function of energy operator. The use of this result with respect to several time-axes gives then similar bounds on the operator without annihilation/creation splitting, and with a~spectral condition needed only in a~neighborhood of zero four-momentum transfer. The proof of Theorem~\ref{Bk} uses a~well-known lemma by Buchholz, which we write down here for the convenience of the reader.
\begin{lem}[\cite{bu90}]\label{bulem}
Let $C$ be a~bounded operator and $P$ the orthogonal projection operator onto the kernel space of $C^n$. Then:
\begin{equation*}
 \|CP\|^2\leq(n-1)\|[C,C^*]\|\,,\quad \|C^*P\|^2\leq n\|[C,C^*]\|\,.
\end{equation*}
\end{lem}

The precise law of decay of commutators mentioned above, on which our analysis is based, is formulated below in Definition~\ref{com} in Section~\ref{bound}. This law takes the form of a~`$\kappa$-type condition' depending on a~decay-rate parameter $\kappa>0$. The law is trivially satisfied for local observables, but it also admits a~wide class of rather strongly nonlocal fields -- their nonlocality is not referred to local observables as is the case in the usual definitions of `quasi-locality'~\cite{ar67} or `almost-locality' (see e.g.\ \cite{ha92}). The precise form of the commutator bounding function $D_\kappa(a)$ may be seen as the least restricting power-law generalizing locality: (i) for $3$-space translation this reduces to $\kappa$-power decay, (ii) the law is Lorentz-covariant, and (iii) for two double-cones of fixed sizes centered at $0$ and $a$ respectively, it is precisely the value of $|\vec{a}|-|a^0|$ which decides whether the regions are spacelike separated.

Our motivation for going outside the paradigm of locality is at least twofold. From a~formal point of view it is worth understanding how far quantum field properties depend on the strict locality assumption. The results of the present paper indicate that at least some of the expected and desirable properties do safely without it. On the other hand, on the physical side, strict locality may prove to be too restrictive in theories with constraints, such as electrodynamics: we note that the charge and infrared structure of QED is still far from being completely understood. We are among the authors believing that proper inclusion of the long-range structure into QED demands the introduction of nonlocal observables. The need for nonlocality may be even more justified in case of non-observable fields, such as e.g.\ the gauge potential in physical gauges. See~\cite{he98} for a~(nonlocal) model of asymptotic fields in QED; see also~\cite{he12} for a~more recent argument for nonlocality and some bibliographic remarks. However, we stress once more that our present results do not depend on our more specific motivation and include, in particular, local case.

In Section~\ref{applic} we give some applications of Theorem~\ref{GB} to the analysis of the energy-momentum transfer of operators of the assumed type. One of the intermediate results is an extension and strengthening of the $3$-momentum spectral analysis of local operators by Buchholz~\cite{bu90}. Another tool used for the analysis is a~scaling degree of the Steinmann type~\cite{st71}, but applied to norms of operators, with scaling in momentum space. We shall say more on the motivation and results of this section below, after introducing our notation and recalling some basic facts. Appendices contain a~few lemmas which are needed in the main text, but which may also be of a~more general technical interest. The implications of Theorem~\ref{GB} will also be used (in a~future publication~\cite{du14}) for the extension to the case of massless fields of the methods used in~\cite{he13} for the analysis of scattering and particle structure in quantum field theory.

We denote by $\M$ the affine Minkowski space built on the Minkowski vector space $M$; the origin $O$ in $\M$ is fixed. The unit, future-pointing vector of a~chosen time-axis is denoted by $t$; $\vx$ is the $3$-space part of the vector $x$, and $|x|^2=|x^0|^2+|\vx|^2$. The `momentum space' isometric with $M$ will be denoted by $\hM$. Throughout the article $\la$ is a~fixed parameter of the physical dimension of length.

We shall use the following conventions and notation for Fourier transforms
\begin{equation*}
\begin{gathered}
 \ti{f}(\w)=(\F_1f)(\w)=\frac{1}{2\pi}\int e^{i\w t} f(t)dt\,,\\
 (\F_3{g})(\vp)=\frac{1}{2\pi}\int e^{-i\vp\cdot\vx} g(\vx)d^3x\,,\\
 \wh{\chi}(p)=(\F{\chi})(p)=\frac{1}{(2\pi)^2}\int e^{ip\cdot x} \chi(x)dx\,,\quad
 \wch{\vph}=\F^{-1}\vph\,.
\end{gathered}
\end{equation*}

We assume that there is a~continuous unitary representation of translations $U(x)=\exp(ix\cdot P)$ acting in a~Hilbert space $\Hc$, with the spectrum contained in $\ov{V_+}$, the closure of the future lightcone. However, we do \emph{not} assume the existence of a~vacuum vector. We denote by $P(\W)$ the projection onto the subspace with spectral values of $P^0$ in $\W$; in particular, we write $P(E)=P((-\infty,E\>)$, whereby $P(E)=0$ for $E<0$.

For each bounded operator $B$ acting in $\Hc$ and an integrable function $\chi$ on $M$ one denotes
\begin{equation*}
B(x)=U(x)BU(-x)\,,\quad B(\chi)=\int B(x)\chi(x)\,dx\,.
\end{equation*}
We extend the second definition to the case when $\nu$ is a~finite complex Borel measure on $M$, and denote
\begin{equation*}
 B(\nu)=\int B(x)\, d\nu(x)=\wch{B}(\wh{\nu})\,,
\end{equation*}
where
\begin{equation*}
 \wh{\nu}(p)=\frac{1}{(2\pi)^2}\int e^{ip\cdot x}\,d\nu(x)
\end{equation*}
is the Fourier transform of the measure. Recall that for each such $\nu$ its variation $|\nu|$ is a~finite positive measure and $|\int\chi d\nu|\leq\int|\chi|d|\nu|$. We note also that
\begin{equation*}
 B(\nu)^*=(B^*)(\bar{\nu})\,.
\end{equation*}
We shall write $B\in\C^n$ (resp.\ $B\in\C_t^n$) if all derivatives $D^\al B(x)$ with $|\al|\leq n$ (resp. $\p_0^lB(x^0)$ with $l\leq n$, in the Minkowski basis in which $t$ is the timelike basis vector) exist and are continuous in the norm sense. All these classes (including $\C^\infty$ and $\C^\infty_t$) are  weakly dense in $\B(\Hc)$.

If the spectral energy-momentum content of a~vector $\psi\in\Hc$ is contained in $\Delta_\psi$ and the support of $\wh{\nu}$ is in $\Delta_{\wh{\nu}}$, then the spectral content of $B(\nu)\psi$ is contained in $\Delta_\psi+\Delta_{\wh{\nu}}$. Therefore, $\wch{B}(p)$ is interpreted as the spectral component of $B$ transferring energy-momentum $p$~\cite{ha92}. After recalling this, we can now more fully discuss the background and announce the results of Section~\ref{applic}.

The most interesting aspect of the spectral properties of translations automorphism is the analysis of singularities, which arise when Fourier transform of $B(x)$ has non-integrable components. This problem is much more involved than the spectral theory of unitary translation group for two main reasons: (i) lack of orthogonality (in contrast to the case of Hilbert space) and issuing difficulties in decomposing the spectral objects, and (ii) the nature of $\wch{B}(p)$ is distributional, in contrast to the more specific situation in Hilbert space, where spectral objects are Borel measures. The first problem was recently addressed by Dybalski in~\cite{dy10}, where also a~general definition of continuous and absolutely continuous spectrum was proposed and a~number of decomposition theorems were proved. Most of the applications to QFT discussed in the article either assume a~vacuum vector, which allows the lifting of the unitary spectral properties to some spectral properties of automorphisms (e.g., as one wold expect, the mass hyperboloid of the massive particle theory is included in the singular continuous spectrum of translation automorphisms), or concentrate on the $3$-space (without time) restriction of the group. The analysis of the latter is of central importance for the derivation of the so called asymptotic functionals, which play central role in the discussion of particle content of the theory according to ideas of Araki and Haag~\cite{ar67} and Buchholz~\cite{bu91},~\cite{bu94}. In its most far-reaching interpretation, this approach aims at replacing Wigner concept of a~particle, which in its standard form is not applicable to objects such as infraparticles, by some decompositions of asymptotic functionals into plane-wave type objects~\cite{po04}.

However, whereas the ideas of Araki, Haag and Buchholz may be a~useful tool for discovering the manifestations of particle-type structures in a~theory, the question whether there are some more fundamental structural properties of the theory which are revealed in this way is still valid; at least for the Wigner particles this is the case. We propagate the view that in absence of vacuum, one should look for mass-hyperboloid-like structures in the energy momentum transfer of fields. We believe that at least some infraparticle-like structures may have such form; this has been tested in a~model of asymptotic electrodynamics~\cite{he13}.

Therefore, the distributional structure of the full space-time Fourier transform $\wch{B}(p)$ is what we are most interested in (and not so much, at least at this stage, in spectral decomposition). And as singularities of $\wch{B}(p)$ could, \emph{a priori}, be sharper than measures (as mentioned in (ii) above), a~general question is: how sharp they can be? Section~\ref{applic} gives some answers to this question. In particular, Proposition~\ref{homsurf} says  that if the decay of the commutator function $D_\kappa$ is not to slow ($\kappa>2$), then (in slightly lose terms) the only possible distributional energy-momentum transfer concentrated on a~submanifold in $\hM$ is represented by a~Dirac delta on a~hypersurface (a few other possibilities occur for slower decay). In the case of mass hyperboloid $p^2=m^2$, $p^0>0$, such possibility in absence of vacuum was discussed in~\cite{he13}.

\section{Estimates}\label{bound}

Let the function $f$ on $\mR$ be smooth and bounded together with all its derivatives. For all real $k>0$ the distributions $(t\pm i0)^{-k-1}$ may be convoluted with such functions and give then functions of the same type. We shall denote
\begin{equation}\label{ftk}
 f^k_\pm(\tau)=\int T_\pm^k(\tau-s)f(s)\,ds\,,\quad T^k_\pm(s)=\mp ie^{\mp ik\pi}\frac{\Gamma(k+1)}{2\pi(s\mp i0)^{k+1}}\,.
\end{equation}
Their distributional Fourier transforms are then
\begin{equation}\label{fwk}
 \ti{f^k_\pm}(\w)=e^{\mp ik\pi/2}\theta(\pm\w)\,|\w|^k\,\ti{f}(\w)\,.
\end{equation}
In particular, for $k=n=1,2,\ldots$ the $n$-th derivative of $f$ is
\begin{equation*}
 f^{(n)}=f^n_++f^n_-\,.
\end{equation*}

For $B\in C_t^\infty$ we shall denote by $B^k_\pm$ the smooth bounded operators obtained by `reflected' convolution \eqref{ftk} applied in time variable,
\begin{equation}
 B^k_\pm(\tau,\vec{0})=\int B(s,\vec{0})T_\pm^k(s-\tau)\,ds\,,
\end{equation}
so that for the translated operator $B^k_\pm(x)$ smeared with a~test function $\vph$ on $M$ there is
\begin{equation*}
 B^k_\pm(\vph)=B(\vph^k_\pm)=\wch{B^k_\pm}(\wh{\vph})=\wch{B}(\wh{\vph^k_\pm})
\end{equation*}
 and
\begin{equation}
 \wch{B^k_\pm}(p)=e^{\mp ik\pi/2}\theta(\pm p^0)\,|p^0|^k\,\wch{B}(p)\,,\quad k>0\,.
\end{equation}

\begin{dfn}\label{com}
We shall say that the commutator $[B_1,B_2]$ of bounded operators \mbox{$B_1,B_2$} is of $\kappa$-type, $\kappa>0$, if the following bound is satisfied:
\begin{equation}\label{fcom}
\begin{aligned}
 &\|[B_1,B_2(a)]\|\leq c D_{\kappa}(a)\,,\\
 &D_{\kappa}(a)\equiv \begin{cases}1 & a^2\geq0\\
                          \dfrac{\la^\kappa}{(\la+|\vec{a}|-|a^0|)^{\kappa}} & a^2<0\,.
             \end{cases}
\end{aligned}
\end{equation}
with some constant $c$ depending on $B_i$. The assumption is covariant: if the bound holds in any particular reference system, it is valid in all other, with some other constants $c$.

We shall say that $[B_1,B_2]$ is of $\kappa^\infty$-type (resp.\ $\kappa_t^\infty$-type), if \mbox{$B_i\in\C^\infty$} (resp.\ $B_i\in\C_t^\infty$) and all $[D^{\al_1}B_1,D^{\al_2}B_2]$  (resp.\ all $[\p_0^{n_1}B_1,\p_0^{n_2}B_2]$) are of $\kappa$-type.
\end{dfn}

In the following lemma we write $B(f)=\int B(x^0,\vec{0})f(x^0)dx^0$.
\begin{lem}\label{stability}
 Let $[B_1,B_2]$ be of $\kappa$-type.\\
 (i) If $|\chi(x)|\leq\con\,(\la+|x|)^{-r}$, $r\geq 4+\kappa$, then $[B_1,B_2(\chi)]$ is also of $\kappa$-type.\\
 (ii) If $|f(\tau)|\leq\con\,(\la+|\tau|)^{-s}$, $s\geq 1+\kappa$, then $[B_1,B_2(f)]$ is also of $\kappa$-type.
\end{lem}
\begin{proof}
 (i) We have $\|[B_1,B_2(\chi)(x)]\|\leq \con\int D_\kappa(z+x)(\la+|z|)^{-r}dz$. We~split integration region into (a) $|z|\leq(|\vx|-|x^0|)/4$, and (b) the rest. In region~(b) we use $D\leq1$, which is sufficient for the bound. In region (a) there is $|\vx+\vec{z}|-|x^0+z^0|\geq(|\vx|-|x^0|)/2$, which again leads to the bound. The proof of (ii) is similar.
\end{proof}
In view of the above lemma, a~shift from $\kappa$-type to $\kappa^\infty$-type or $\kappa^\infty_t$-type may be easily achieved by smearing:
\begin{lem}
 Let $[A_1,A_2]$ be of $\kappa$-type. If $B_i=A_i(\chi_i)$, where all $D^\al\chi_i$ satisfy the assumptions of the last lemma, then $[B_1,B_2]$ is of $\kappa^\infty$-type. Similarly, if $C_i=A_i(f_i)$, where all $f_i^{(n)}$ satisfy the assumptions of the last lemma, then $[B_1,B_2]$ is of $\kappa^\infty_t$-type.
\end{lem}

An immediate consequence for operators satisfying Definition~\ref{com} is the following bound:
\begin{equation}\label{BBD}
 \|[B_1(\nu_1),B_2(\nu_2)]\|\leq c\int D_\kappa(x-y)\,d|\nu_1|(x)\,d|\nu_2|(y)\,.
\end{equation}

Our main technical result is the following theorem.
\begin{thm}\label{Bk}
Let $[B,B^*]$ be of $\kappa^\infty_t$-type and let $E\geq0$. \\
(i) If $k>(\kappa+1)/2$, then
\begin{equation}\label{Bkb}
 \|B^k_\pm(\nu)P(E)\|\leq \con\, \bigg\{c_\pm(E)\int D_\kappa(x-y)\,d|\nu|(x)\,d|\nu|(y)\bigg\}^{1/2}\,,
\end{equation}
where $c_+(E)=1+\la E$, $c_-(E)=\la E$, and the bounding constant depends on $B$, $k$ and $\kappa$.\\
(ii) If $k\geq\kappa$, then also $[B_1,B^k_{2\pm}]$ and $[B^k_{1\pm},B^k_{2\pm}]$ (uncorrelated signs) are of $\kappa$-type.
\end{thm}
\begin{proof}
(i) We start with some function analysis. Let $\ti{\eta}(\w)$ be a~smooth  even function with support in $\<-\la^{-1},\la^{-1}\>$ and equal to $1$ on $\<-(2\la)^{-1},(2\la)^{-1}\>$. We form $\eta^k_+$ as in \eqref{ftk} and $\eqref{fwk}$. Then the use of Lemma~\ref{decay} in Appendix~\ref{dec} gives $|\eta^k_+(s)|\leq\con\,(\la+|s|)^{-1-k}$. Introduce smooth functions on~$\mR$
\begin{equation*}
 \ti{j}(\w)=\theta(\w)[\ti{\eta}(\w)-\ti{\eta}(2\w)]\,,\quad \ti{j_n}(\w)=\ti{j}(2^n\w)=\theta(\w)[\ti{\eta}(2^n\w)-\ti{\eta}(2^{n+1}\w)]
\end{equation*}
with $\supp{\ti{j}}\subset\<(4\la)^{-1},\la^{-1}\>$, $\supp{\ti{j_n}}\subset\<(2^{n+2}\la)^{-1},(2^n\la)^{-1}\>$.
It is then easy to find that
\begin{gather}
 \eta^k_+(s)-\sum_{n=0}^{N-1}j^k_{n+}(s)=\frac{1}{2^{N(k+1)}}\eta^k_+(2^{-N}s)\,,\nonumber \\
 \Big\|\eta^k_+-\sum_{n=0}^{N-1}j^k_{n+}\Big\|_1=\frac{1}{2^{Nk}}\|\eta^k_+\|_1\rightarrow 0\quad(N\rightarrow\infty)\,.\label{etanorm}
\end{gather}

Let now $B$ be a~smooth bounded operator and $k>0$. We define operators $B^k_>$, $B^k_<$ and $B^k_n$ ($n=0,1,\ldots$) by
\begin{equation*}
 \wch{B^k_>}(p)=(1-\ti{\eta}(p^0))\wch{B^k_+}(p)\,,\
 \wch{B^k_<}(p)=\ti{\eta}(p^0)\wch{B^k_+}(p)\,,\
 \wch{B^k_n}(p)=\ti{j_n}(p^0)\wch{B^k_+}(p)\,.
\end{equation*}
As $\ti{\eta}$ and $\ti{j_n}$ are Schwartz functions, all these operators are well defined and smooth, and their energy transfers are in $\<(2\la)^{-1},+\infty)$, $\<0,\la^{-1}\>$ and $\<(2^{n+2}\la)^{-1},(2^n\la)^{-1}\>$, respectively. Moreover, for $l>k+1$ we write
\begin{equation*}
 \wch{B^k_>}(p)=\ti{f_+}(p^0)(1+\la^l(p^0)^l)\wch{B}(p)\,,\quad \ti{f_+}(\w)=e^{-ik\pi/2}\theta(\w)\w^k\frac{1-\ti{\eta}(\w)}{1+\la^l\w^l}
\end{equation*}
and note that $f_+(\tau)$ decays rapidly for $|\tau|\to\infty$. In configuration space we have
\begin{gather}
 B^k_>(x)=\frac{1}{2\pi}\int C(x^0+\tau,\vx)f_+(\tau)\,d\tau\,,\quad C=\big(1+(-i\la\p_0)^l\big)B\,,\label{B>}\\
 B^k_<(x)=\frac{1}{2\pi}\int B(x^0+\tau,\vx)\,\eta^k_+(\tau)\,d\tau\,,\label{B<}\\
 B^k_n(x)=\frac{1}{2\pi}\int B(x^0+\tau,\vx)\,j^k_{n+}(\tau)\,d\tau\,,\quad j^k_{n+}(\tau)=2^{-n}j^k_+(2^{-n}\tau)\,.\label{Bnj}
\end{gather}
By assumption $[C,C^*]$ is of $\kappa$-type, so also $[B^k_>,B^{k*}_>]$ is of $\kappa$-type. Moreover, using Eq. \eqref{etanorm} we obtain
\begin{equation}\label{BBn}
 \Big\|B^k_<-\sum_{n=0}^{N-1}B^k_n\Big\|\leq\frac{1}{2^{Nk}}\|\eta^k_+\|_1\|B\|\rightarrow 0\,.
\end{equation}

Taking into account the scope of energy transfers of operators $B^k_>{}^*$ and $B^k_n{}^*$ and using Lemma~\ref{bulem} we have
\begin{gather*}
 \|B^k_>(\nu)^\#P(E)\|^2\leq d(\la E)^\#\int\|[B^k_>{}^*,B^k_>(x-y)]\|\,d|\nu|(x)\,d|\nu|(y)\,,\\
 \|B^k_n(\nu)^\#P(E)\|^2\leq d_n(\la E)^\#\int\|[B^k_n{}^*,B^k_n(x-y)]\|\,d|\nu|(x)\,d|\nu|(y)\,,
\end{gather*}
where $d(\sigma)=2\sigma+1$, $d(\sigma)^*=2\sigma$, $d_n(\sigma)=2^{n+2}\sigma+1$, $d_n(\sigma)^*=2^{n+2}\sigma$.
From the first of the above relations the thesis follows for the part $B^k_>$ of $B^k_+$. For $B^k_n$ we use \eqref{Bnj} and \eqref{BBD}, which leads to
\begin{multline*}
 \|[B^k_n{}^*,B^k_n(a)]\|\leq \con\int D_\kappa(a^0+\tau,\vec{a})\,|j^k_{n+}(\tau+s)||j^k_{n+}(s)|dsd\tau\\
 =\con\, 2^{-n(2k+1)}\int D_\kappa(a^0+\tau,\vec{a})J_k(2^{-n}\tau)d\tau\,,
\end{multline*}
where $J_k(\tau)=\int|j^k_+(\tau+u)||j^k_+(u)|du$ is a~function of fast decrease.
Therefore for each $N$ we  have the bound
\begin{equation*}
 \|[B^k_n{}^*,B^k_n(a)]\|\leq\frac{\con}{2^{n(2k+1)}}\int \frac{D_\kappa(a^0+\tau,\vec{a})}{(\la+|\tau|/2^n)^{N+1}}d\tau\,.
\end{equation*}
For $a^2\geq0$ we use $D_\kappa\leq 1$ and find that the rhs is bounded by $\con/2^{n2k}$. For $a^2<0$ we consider two integration regions separately: (i) $|\tau|\geq(|\vec{a}|-|a^0|)/2$, and (ii) the rest. In the second region $|\vec{a}|-|a^0+\tau|>(|\vec{a}|-|a^0|)/2$, so this contribution is bounded by $\con\,2^{-n2k}D_\kappa(a)$. In the first region we use \mbox{$D_\kappa\leq1$} and then this contribution is bounded by \mbox{$\con\,2^{-n2k}[\la+(|\vec{a}|-|a^0|)/2^n]^{-N}$}. Take $N\geq\kappa$, then
\begin{equation*}
 [\la+(|\vec{a}|-|a^0|)/2^n]^{-N}\leq\con\,[...]^{-\kappa}\leq\con\, 2^{n\kappa}(\la+|\vec{a}|-|a^0|)^{-\kappa}\,.
\end{equation*}
Summing up, we obtain
\begin{gather*}
 \|[B^k_n{}^*,B^k_n(a)]\|\leq\frac{\con}{2^{n(2k-\kappa)}}D_\kappa\,(a)\,,\\[1.5ex]
 \|B^k_n(\nu)^\#P(E)\|^2\leq\frac{\con\, d_n(E/E_0)^\#}{2^{n(2k-\kappa)}}\int D_\kappa(x-y)|\,d|\nu|(x)\,d|\nu|(y)\,.
\end{gather*}
Finally, taking into account \eqref{BBn} we have
\begin{equation*}
 \|B^k_<(\nu)^\#P(E)\|\leq\sum_{n=0}^\infty\|B^k_n(\nu)^\#P(E)\|\,,
\end{equation*}
which for $2k-\kappa-1>0$ is summable and leads to the result (i) of the theorem.\\
(ii) We split $B^k_{i+}=B^k_{i>}+B^k_{i<}$ and similarly for $B^k_{i-}$. The use of representations \eqref{B>} and \eqref{B<} and their analogies for $B^k_{i-}$ together with Lemma~\ref{stability} leads to the thesis.
\end{proof}

We shall denote by $G_\pm(E)$ any real functions $\<0,+\infty)\mapsto\<0,+\infty\>$ which satisfy the following conditions:
\begin{equation}\label{GE}
\begin{split}
 \text{(i)}\quad &G_\pm(E)\ \text{are non-increasing}\,,\\
 \text{(ii)}\quad &G_\pm\in L^2(\<0,+\infty))\,,\\
 \text{(iii)}\quad  &G_+(E)\leq\con\,
\end{split}
\end{equation}
(note that $G_-(0)$ may take the value $+\infty$). Moreover, we set
\begin{equation}\label{defBk}
 \wch{B^k_t}(p)=|p|^k\wch{B}(p)\,.
\end{equation}
To simplify notation we also identify $G_t=G_+$.

\begin{thm}\label{GB}
Let $G(E)$ be as defined above and let $k>(\kappa+1)/2$. \\
(i) If $[B,B^*]$ is of $\kappa^\infty_t$-type, then
\begin{equation}
 \|B^k_\pm(\nu)G_\pm(P^0)\|^2\leq\con\int D_\kappa(x-y)\,d|\nu|(x)\,d|\nu|(y)\,.\label{boundpm}
\end{equation}
(ii) If $[B,B^*]$ is of $\kappa^\infty$-type, then the expression \eqref{defBk} is the Fourier transform of a~bounded operator $B^k_t$, and
\begin{equation}
 \|B^k_t(\nu)G_t(P^0)\|^2\leq\con\int D_\kappa(x-y)\,d|\nu|(x)\,d|\nu|(y)\,.\label{boundt}
\end{equation}
In both cases the bounding constants depend on $B$, $\kappa$, $k$ and $G$.
\end{thm}
\begin{proof}
(i) We put $\vep=\pm$ and extend $G_\vep(E)=G_\vep(0)$ for $E<0$. For each $\psi\in\Hc$ we have the spectral representation
\begin{equation*}
 \|G_\vep(P^0)B^k_\vep(\nu)^*\psi\|^2=\int_{\mR}G^2_\vep(E)\,d\|P(E)B^k_\vep(\nu)^*\psi\|^2\,,
\end{equation*}
where $G^2_\vep(E)=[G_\vep(E)]^2$. Let $J$ denote the integral on the rhs of \eqref{boundpm}. Then by Theorem~\ref{Bk} we have $\|P(E)B^k_\vep(\nu)^*\psi\|^2\leq \con\,\|\psi\|^2J\, c_\vep(E)$, where $c_\ep$ are extended to $\mR$ by $c_\ep(E)=0$ for $E<0$. But then $c_\vep(E)=\mu_\vep((-\infty,E\>)$, where $d\mu_+(E)=(\delta(E)+\la\theta(E))\,dE$ and $d\mu_-(E)=\la\theta(E)\,dE$, so by Lemma~\ref{mumu} in Appendix~\ref{integrals} we obtain
\begin{equation*}
 \|B^k_\vep(\nu)G_\vep(P^0)\|^2\leq\con\bigg\{\delta_{\vep,+}G^2_+(0)+\la\int_0^\infty G^2_\vep(E)\,dE\bigg\}\,J\,.
\end{equation*}
This ends the proof of the inequality \eqref{boundpm}.\\
(ii) We choose $k'>(\kappa+1)/2$ such that $1>k-k'>0$. We also choose a~basis $\{t_i\}_{i=1}^4$ of timelike, unit, future-pointing vectors. Let $n$ be a~positive integer such that $2n-k'\geq5$. Then for multi-indices $\al$ with $|\al|\leq5$ the functions
\begin{equation*}
 \wh{\chi_i}(p)=\frac{|p|^k\,|t_i\cdot p|^{2n-k'}}{(1+\la^2|p|^2)^3\sum_{j=1}^4[t_j\cdot p]^{2n}}
\end{equation*}
satisfy outside $p=0$ the bounds $|D^\al\wh{\chi_i}(p)|\leq\con\,|p|^{k-k'-|\al|}(1+\la^2|p|^2)^{-3}$, so the assumptions of Lemma~\ref{decay} in Appendix~\ref{dec} are satisfied with $\ga=k-k'$. Therefore, both $\wh{\chi_i}$ and $\chi_i$ are integrable. Moreover, we have
\begin{equation*}
 |p|^k=\sum_{i=1}^4\wh{\chi_i}(p)|t_i\cdot p|^{k'}(1+\la^2|p|^2)^3\,.
\end{equation*}
We denote $C=\big(1-\la^2(\p_0^2+\triangle)\big)^3B$. Then
\begin{equation*}
 \wch{B^k_t}(p)=\sum_{i=1}^4\wh{\chi_i}(p)|t_i\cdot p|^{k'}\,\wch{C}(p)=\sum_{i=1}^4\wh{\chi_i}(p)\sum_{\ep=\pm}e^{\ep ik'\pi/2}\wch{C^{k'}_{i\ep}}(p)\,,
\end{equation*}
where $C^{k'}_{i\pm}$ are defined in analogy to $C^{k'}_\pm$, but with respect to the time axes~$t_i$. This shows that $B^k_t$ is defined as a~bounded operator. Next, we note that if $t_i\cdot t=\cosh(\xi_i)$, then $t\cdot p\geq e^{-\xi_i}t_i\cdot p$ for $p\in\ov{V_+}$. Therefore $G_+(t\cdot p)\leq G_+(e^{-\xi_i}t_i\cdot p)\equiv G_{i+}(t_i\cdot p)$. The lhs of the inequality \eqref{boundt} is then bounded by $\dsp\con\sum_{i}\sum_{\ep=\pm}\|C^{k'}_{i\ep}(\chi_i)G_{i+}(t_i\cdot P)\|^2$; inequality \eqref{boundpm} applies to $C^{k'}_{i\pm}$, so by integrability of $\chi_i$ also to $C^{k'}_{i\pm}(\chi_i)$, which leads to the thesis.
\end{proof}

\section{Applications to energy-momentum transfer}\label{applic}

Before entering a~more extensive discussion with the use of further tools we note two immediate consequences of Theorem~\ref{GB}.
\begin{col}\label{point}
Let $\vph$ be an integrable function on $M$, $\ga\geq1$ and $\delta\in(0,1)$. Denote
\begin{equation*}
\begin{aligned}
 \wh{\vph_{q,\ga}}(p)&=\wh{\vph}(\ga^\delta(p^0-q^0),\ga(\vp-\vec{q}))\,,\\ \text{so}\quad
 \vph_{q,\ga}(x)&=\ga^{-3-\delta}e^{-iq\cdot x}\vph(\ga^{-\delta}x^0,\ga^{-1}\vx)\,.
\end{aligned}
\end{equation*}
Then for $B$ with $[B,B^*]$ of $\kappa^\infty$-type, $\kappa>0$, for each $k>1/2$ and all $q\in\hM$ there is
\begin{equation*}
 \lim_{\ga\to\infty}\|\wch{B^k_t}(\wh{\vph_{q,\ga}})G_t(P^0)\|=\lim_{\ga\to\infty}\|B^k_t(\vph_{q,\ga})G_t(P^0)\|=0\,.
\end{equation*}
\end{col}
\begin{proof}
 For $k>1/2$ choose $0<\kappa'<\min\{2k-1,\kappa\}$. Then the use of Theorem~\ref{GB}, upon a~change of integration variables, yields
\begin{equation*}
 \|B^k_t(\vph_{q,\ga})G_t(P^0)\|^2\leq\con \int D_{\kappa'}(\ga^\delta(x^0-y^0),\ga(\vx-\vy))|\vph(x)\vph(y)|\,dxdy\,.
\end{equation*}
The function $|\vph(x)\vph(y)|$ is integrable, and the remaining factor function in the integrand is bounded and tends point-wise almost everywhere to zero for $\ga\to\infty$, which ends the proof.
\end{proof}

Another simple result of similar type is the following.

\begin{col}\label{plane}
 Let $B$, $\kappa$ and $k$ be as in the last Corollary. Let $n$ be a~unit spacelike vector and for $p\in\hM$ decompose $p=p^n\, n+p_\bot$, where $n\cdot p_\bot=0$. Let $\vph$ be an integrable function on $M$ and denote $\wh{\vph_{n,r,\ga}}(p)=\wh{\vph}(\ga(p^n-r)n+p_\bot)$. Then
\begin{equation*}
 \lim_{\ga\to\infty}\|\wch{B^k_t}(\wh{\vph_{n,r,\ga}})G_t(P^0)\|=\lim_{\ga\to\infty}\|B^k_t(\vph_{n,r,\ga})G_t(P^0)\|=0\,.
\end{equation*}
\end{col}
\begin{proof}
 Choose a~Minkowski frame in which $\wh{\vph_{n,r,\ga}}(p)=\wh{\vph}(p^i,\ga(p^3-r))$, $i\neq3$, and then the result follows as in the last Corollary.
\end{proof}

We note that the form of $\wh{\vph_{q,\ga}}$ in Corollary~\ref{point} shows that the limit considered there tests point properties of $\wch{B^k_t}(p)G_t(P^0)$. In particular, it shows that this operator cannot have a~distributional component supported at one point (we shall return to this point below). At the same time this limit is a~generalization of a~weak limit used by Dybalski~\cite{dy10} in his definition of continuous spectrum of translation automorphisms (where $\vph$ is a~characteristic function of a~cuboid). Similarly, Corollary~\ref{plane} shows that $\wch{B^k_t}(p)G_t(P^0)$ cannot have a~distributional component with support on any part of any timelike hyperplane. The previous Corollary is a~particular case of the above.

In order to derive other results of similar physical interpretation we develop further tools. Let $d\nu(x)=\delta(x^0-\tau)f(\vx)\,dx$, $f\in\Sc$, and denote $B(\nu)=B(\tau,f)$. Then the bounds in Theorem~\ref{GB} take the form
\begin{equation}\label{Bff}
 \|B^k_\ep(\tau,f)G_\ep(P^0)\|^2 \leq\con\int |f(\vx)|D_\kappa(0,\vx-\vy)|f(\vy)|\,d^3xd^3y\,,
\end{equation}
where $\ep=\pm,t$. We shall denote by $\|.\|_p$ the $L^p$-norm (with the $d^3x$ or $d^4x$ measure). An immediate consequence of this bound is a~generalization and sharpening of a~result due to Buchholz~\cite{bu90}.
\begin{pr}\label{Bfbound}
Let $[B,B^*]$ be of $\kappa^\infty$-type, $\kappa>0$. Then  the following bound holds
\begin{equation}\label{Bf3}
 \|B^k_\ep(\tau,f)G_\ep(P^0)\|\leq\con\|f\|_p\,,\quad \ep=\pm,t\,,
\end{equation}
where
\begin{align*}
    &k>(\kappa+1)/2, & &p=6/(6-\kappa) & \text{for}\ &\kappa<3,\\
    &k>2, & &p=2 & \text{for}\ &\kappa>3,
\end{align*}
with the bounding constant depending on $B$, $\kappa$, $k$ and $G_\vep$.
\end{pr}
\begin{proof}
For $\kappa<3$ and $k>(\kappa+1)/2$ the bound \eqref{Bff} applies. In this case we observe that $D_\kappa(0,\vec{z})<|\vec{z}|^{-\kappa}$ and then the application of the Sobolev inequality (see e.g.~\cite{rs79}) leads to the bound in the thesis with \mbox{$p=6/(6-\kappa)$}. For $\kappa>3$ and any $k>2$ the bound \eqref{Bff} applies with any $\kappa'\in(3,\kappa\>$ such that $\kappa'<2k-1$. In this case we consider the integral operator $K$ with the kernel $K(\vx-\vy)=(\la+|\vx-\vy|)^{-\kappa'}$. As $K(\vx)$ is an integrable function, in momentum representation this operator is the multiplication by a~bounded function. Thus $K$ is a~bounded operator in $L^2$, so $(|f|,K|f|)\leq\con\|f\|^2_2$, which ends the proof.
\end{proof}

We denote by $X$ the operator $(X\vph)(x)=x\vph(x)$.
\begin{pr}\label{Bfibound}
Let $[B,B^*]$ be of $\kappa^\infty$-type, $\kappa>0$. Then the following bound holds
\begin{equation}\label{Bfi}
 \begin{aligned}
 \|B^k_\ep(\vph)G_\ep(P^0)\|&\leq\con\int\|\vph(x^0,.)\|_p\,dx^0\\
            &\leq\con\,\|(\la+|X^0|)^\sigma\vph\|_p\,,\quad \ep=\pm,t\,,
 \end{aligned}
\end{equation}
where
\begin{align*}
    &k>(\kappa+1)/2, & &p=6/(6-\kappa), & &\sigma>\kappa/6 & &\text{for}\ \kappa<3,\\
    &k>2, & &p=2, & &\sigma>1/2 & &\text{for}\ \kappa>3,
\end{align*}
with the bounding constant depending on $B$, $p$, $k$, and $G$ (and $\sigma$ in the second form). In particular, for $\kappa>3$, $k>2$ and $\sigma>1/2$:
\begin{equation}\label{Bvp}
 \|\wch{B^k_\ep}(\wh{\vph})G_\ep(P^0)\|\leq\con\Big(\la^{2\sigma}\|\wh{\vph}\|^2_2+\|\p^\sigma_0\wh{\vph}\|^2_2\Big)\,.
\end{equation}
For $\kappa<3$, $k>(\kappa+1)/2$ one also has
\begin{multline}\label{Bfitwo}
 \|B^k_\ep(\vph)G_\ep(P^0)\|\leq\con\int\|(\la+|\vec{X}|)^\beta\vph(x^0,.)\|_2\,dx^0\\
 \leq\con\big\|(\la+|X^0|)^\tau(\la+|\vec{X}|)^\beta\vph\big\|_2\,,
\end{multline}
where $\beta>(3-\kappa)/2$ and $\tau>1/2$.
\end{pr}
\begin{proof}
 The first inequality in \eqref{Bfi} is an immediate consequence of the bound \eqref{Bf3}, and for the second we write the integrand of the first bound as
\begin{equation*}
 f(x^0)g(x^0)=(\la+|x^0|)^{-\sigma}\times\big[(\la+|x^0|)^\sigma\|\vph(x^0,.)\|_p\big]
\end{equation*}
and use H\"older's inequality $\|fg\|_1\leq\|f\|_{p/(p-1)}\|g\|_p$. For $p=2$ the inequality \eqref{Bvp} follows by Fourier transform. Similarly, both inequalities in \eqref{Bfitwo} follow by H\"older's inequality.
\end{proof}

\begin{rem}
 If $[B,B^*]$ is of $\kappa_0^\infty$-type, then Proposition~\ref{Bfbound} and Proposition~\ref{Bfibound} apply with all $\kappa\in(0,\kappa_0\>$. In particular: the proofs of these propositions show that $\kappa_0=3$ is a~critical value, but in this case they hold with all $\kappa<3$.
\end{rem}

As a~tool for further investigation of singularities in energy-momentum transfer we introduce the following analogues of Steinmann's scaling degree of distributions~\cite{st71}.\footnote{Our definition is applied to `operator distributions' defined on the whole Minkowski space, thus we do not have to face problems of the admissibility of extension, as in the discussion of scaling degrees defined by Brunetti and Fredenhagen~\cite{br00}. Moreover, we use the convention of the original Steinmann's definition rather than the version used later by other authors, which differs by sign. Steinmann's definition gives the usual homogeneity degree for homogeneous distributions. We note also, that the context in which we use the scaling degree is `complementary' to the usual use in renormalization theory. There it serves to control short distance -- high energy-momentum behavior, while we use it to test local energy-momentum singularities.}
\begin{dfn}\label{scdegsigma}
 Let $\Sigma\subseteq\hM$ be a~smooth submanifold of codimension $m$, and $q\in\Sigma$. Choose any smooth local coordinate system $(\rho,\sigma)=(\rho^i,\sigma^j)$, \mbox{$i=1,\ldots,m$}, $j=1,\ldots,4-m$, in an open neighborhood $U$ of $q$  such that $\Sigma$ is the solution of $\rho^i=0$, \mbox{$i=1,\ldots,m$}. Let $\wh{\vph}(p)=\psi(\rho,\sigma)$, where $\psi$ is smooth and of compact support, such that $\supp\wh{\vph}\subseteq U$. Denote
\begin{equation*}
 \wh{\vph_\ga}(p)=\ga^m\psi(\ga\rho,\sigma)\,.
\end{equation*}
For a~bounded operator $B$ we define the momentum scaling degree $d_{q,\Sigma,U}(B)$ as the supremum of the set of all numbers $c$ such that
\begin{equation*}
 \lim_{\ga\to\infty}\ga^c\|B(\vph_\ga)\|=\lim_{\ga\to\infty}\ga^c\|\wch{B}(\wh{\vph_\ga})\|=0
\end{equation*}
for all functions $\wh{\vph_\ga}$ constructed in the above defined way. For $U_1\subseteq U_2$ there is $d_{q,\Sigma,U_1}(B)\geq d_{q,\Sigma,U_2}(B)$, so we can define the scaling index
\begin{equation*}
 d_{q,\Sigma}(B)=\sup_Ud_{q,\Sigma,U}(B)\,.
\end{equation*}
In particular, for $m=4$ we obtain a~scaling degree at a~point, which we denote $d_q(B)$.
\end{dfn}
\begin{rem}\label{scdegq}
If $\nu$ is a~complex Borel measure, then $d_{q,\Sigma}(B(\nu))\geq d_{q,\Sigma}(B)$ for all~$q$ and $\Sigma$.
\end{rem}
\noindent
Remark follows directly from $[B(\nu)](\vph)=[B(\vph)](\nu)$.
\begin{pr}\label{scsigma}
 Let $q\in\Sigma\subset\hM$, with $\Sigma$~a~local smooth submanifold with codimension $m$. \\
(i) If $B\in\B(\Hc)$, then
\begin{equation*}
 d_{q,\Sigma}(B)\geq-(m+4)/2\,.
\end{equation*}
(ii) If $[B,B^*]$ is of $\kappa$-type and $q\neq0$, then
\begin{align*}
 &d_{q,\Sigma}(BG_t(P^0))\geq-(m+4-\kappa)/2\,,& &\kappa<3\,,\\
 &d_{q,\Sigma}(BG_t(P^0))\geq-(m+1)/2\,,& &\kappa\geq3\,.
\end{align*}
\end{pr}
\begin{proof}
If $B\in\B(\Hc)$, then for each $b>2$ we have the estimate
\begin{equation}\label{Bbare}
 \|B(\vph_\gamma)\|\leq\|B\|\|\vph_\gamma\|_1\leq\con\|(\la^b+|X|^b)\vph_\gamma\|_2\,.
\end{equation}
Moreover, if $[B,B^*]$ is of $\kappa$-type and $q\neq0$, then referring to Definition~\ref{scdegsigma} we can assume without restricting generality that $0\notin \ov{U}$. In this case let $k$ be as given by Proposition~\ref{Bfibound}. For $\wh{\vph_\ga}$ with support in $U$ there is $\wch{B}(\wh{\vph_\ga})=\wch{B(h)^k_t}(\wh{\vph_\ga})$, where $\wh{h}$ is a~test function of compact support, equal to $|p|^{-k}$ on~$U$. Proposition~\ref{Bfibound} now gives
\begin{equation}\label{Bh}
 \|\wch{B}(\wh{\vph_\ga})G_t(p^0)\|\leq\con\,\big\|(\la^\tau+|X^0|^\tau)(\la^\beta+|\vec{X}|^\beta)\vph_\ga\big\|_2\,,
\end{equation}
where $\tau>1/2$, $\beta>(3-\kappa)/2$ for $\kappa<3$ and $\beta=0$ for $\kappa>3$.

It is easily seen that one can find a~continuous function of compact support $\Psi(\rho,\sigma)$ such that $|D^\al\wh{\vph_\ga}(p)|\leq\ga^{|\al|+m}\Psi(\ga\rho,\sigma)$ for $|\al|$ in a~finite set ($D^\al$ are with respect to Minkowski coordinates). Therefore,  changing the integration variables from Minkowski $p^i$ to $\rho^i,\sigma^j$ one finds that \mbox{$\|D^\al\wh{\vph_\ga}\|_2\leq\con\,\gamma^{|\al|+(m/2)}$}. The use of Lemma~\ref{inequal} in Appendix~\ref{inequality} now shows that the norms on the rhs of \eqref{Bbare} and \eqref{Bh} are bounded by $\con\,\ga^s$, where $s=(m/2)+b$ for \eqref{Bbare}, and $s=(m/2)+\tau+\beta$ for \eqref{Bh}. Therefore, $s$ is any number $>(m+4)/2$ in case of \eqref{Bbare}. For the bound \eqref{Bh} there are two subcases: if $\kappa<3$, then $s$ is any number $>(m+4-\kappa)/2$, and if $\kappa>3$, then $s$ is any number $>(m+1)/2$. The thesis follows.
\end{proof}

A consequence of this result is the following statement. The scaling degree $\w(T)$ below is the supremum of the set of numbers $c$ such that $\ga^cT(\psi_\gamma)\to0$ for all test functions $\psi_\ga(\rho)=\ga^m\psi(\ga\rho)$ and $\ga\to\infty$, $\psi$ smooth of compact support.
\begin{pr}\label{homsurf}
Let $[B,B^*]$ be of $\kappa$-type and $\eta$ a~norm-continuous linear functional on $\B(\Hc)$. Let $q$, $U$ and $\rho^1,\ldots,\rho^m$ be as denoted in Definition~\ref{scdegsigma}, but here the variables $\rho$ are fixed. Suppose that locally on $U$ there is
\begin{equation*}
 \eta(\wch{B}(p)G_t(P^0))=\wh{f}(p)T(\rho)+\wh{g}(p)\,,
\end{equation*}
where $\wh{f}$ is a~smooth function on $U$, $\wh{f}(q)\neq0$, $\wh{g}\in L^1(U,dp)$, and $T$ is an $m$-dimensional distribution with scaling degree $\w(T)$.
Then
\begin{equation}\label{wT}
 \w(T)\geq\min\{d_{q,\Sigma}(BG_t(P^0)),-m\}\,.
\end{equation}
Moreover, if $-m<d_{q,\Sigma}(BG_t(P^0))$, then
\begin{equation}\label{m0}
 \lim_{\ga\to\infty}\ga^{-m}T(\psi_\ga)=0\,.
\end{equation}
Therefore, the only distributions $T$ concentrated at $\rho^i=0$ which are not forbidden by this theorem are the following:\\
-- Dirac delta at $q=0$ for $m=4$,\\
-- Dirac delta for $m=1$,\\
-- for $\kappa\leq2$: additionally Dirac delta for $m=2$,\\
-- for $\kappa\leq1$: additionally Dirac delta for $m=3$ and the first derivative of delta for $m=1$.
\end{pr}
\begin{proof}
Without restricting generality we can assume that $U$ is small enough to satisfy: $|\wh{f}(p)|\geq\con>0$ on $U$. Let $\wh{h}(p)$ be a~smooth function of compact support, equal to $[\wh{f}(p)|J(\rho,\sigma)|]^{-1}$ on $U$, where $J$ is the Jacobian of transformation from $p^i$ to $(\rho^i,\sigma^j)$, with any $\sigma^j$ as denoted in Definition~\ref{scdegsigma}. Then locally on~$U$
\begin{equation*}
 \eta\big(\wch{B(h)}(p)G_t(P^0)\big)=\frac{T(\rho)}{|J(\rho,\sigma)|}+\wh{F}(p)\,,\quad \wh{F}(p)=\frac{\wh{g}(p)}{\wh{f}(p)|J(\rho,\sigma)|}\,.
\end{equation*}
Let $\wh{\vph_\ga}(p)=\psi_\ga(\rho)\chi(\sigma)$, with $\psi_\ga(\rho)=\ga^m\psi(\ga\rho)$ and $\int\chi(\sigma)d\sigma=1$. Then
\begin{equation}\label{etaTo}
 \eta\big(B(h)(\vph_\ga)G_t(P^0)\big)=T(\psi_\ga)+o(\ga^m)\,,
\end{equation}
where the second term on the rhs results from the estimate
\begin{equation*}
 \ga^{-m}\|\wh{F}\wh{\vph_\ga}\|_1\leq\con\hspace{-.7em} \int\limits_{\ga\rho^i\leq\con}\hspace{-.7em}|\wh{F}(p)|dp\,\longrightarrow0\quad (\ga\to\infty).
\end{equation*}
The estimate $\big|\eta\big(B(h)(\vph_\ga)G_t(P^0)\big)\big|\leq\|\eta\|\|h\|_1\|B(\vph_\ga)G_t(P^0)\|$ now leads to the bound \eqref{wT} for~$\w(T)$. Property~\eqref{m0} also follows from equation~\eqref{etaTo}. If $T$ is concentrated at $\rho=0$, then it is a~finite combination of the derivatives of delta, and $\w(T)=-m-l$, where $l$ is the highest degree of derivative in the combination. Eq.\,\eqref{m0} shows that such distributions are not allowed if $-m<d_{q,\Sigma}(BG_t(P^0))$. Therefore, taking into account Eq.\,\eqref{wT}, we must have $-m-l\geq d_{q,\Sigma}(BG_t(P^0))$. In case when $m=4$ and $q=0$
we refer to Proposition~\ref{scsigma} (i), which excludes $l\geq1$, but leaves Dirac delta, which is the first of the possibilities admitted by the thesis. In all other cases we can assume $q\neq0$: for $m\leq3$ one can shift $q$ slightly over $\Sigma$, if necessary. By Proposition~\ref{scsigma} (ii) we now have $m+2l\leq1$ for $\kappa\geq3$, and $m+2l\leq 4-\kappa$ for $\kappa<3$. This leads to all other cases enumerated in the thesis.
\end{proof}

\section*{Appendix}
\setcounter{section}{0}
\renewcommand{\thesection}{\Alph{section}}

\section{A decay property}\label{dec}

\begin{lem}\label{decay}
 Let $F$ be a~function on $\mR^n\setminus\{0\}$ and $\ga>-n$. Suppose that for  multi-indices $\al$ with $|\al|\leq\ga+n+1$  all $D^\al F$ are measurable functions, $D^\al F\in L^1(|p|\geq\la^{-1})$,  and for $|p|\in(0,\la^{-1})$ satisfy the bounds: $|D^\al F(p)|\leq\con\,|p|^{\ga-|\al|}$. Moreover, we demand $\dsp\lim_{|p|\to\infty}|D^\al F(p)||p|^{n-1}=0$ for $|\al|\leq\ga+n$. Then
\begin{equation*}
 |\wch{F}(x)|\leq\frac{\con}{(\la+|x|)^{n+\ga}}\,.
\end{equation*}
\end{lem}
\begin{proof}
Suppose $\ga\in(-n+l,-n+1+l\>$. Then multiplying $\wch{F}(x)$ by all $x^\al$ with $|\al|=l$ and integrating by parts in standard way (which is possible due to our assumptions) one reduces the problem to the case $\ga\in(-n,-n+1\>$, which we now prove. In this case the assumption on integrability and on behavior in a~neighborhood of zero applies to $|\al|\leq2$ if $\ga=-n+1$, and to $|\al|\leq1$ otherwise; the assumption on the limit in infinity applies to $|\al|\leq1$ and $|\al|=0$ respectively. For $|x|\leq\la$ the transform is bounded by a~constant. Let $|x|>\la$. Then
\begin{equation*}
 \bigg|\int\limits_{|p|\leq|x|^{-1}} F(p)e^{-ip\cdot x}\,d^np\,\bigg|\leq\con\int_0^{|x|^{-1}}|p|^{\ga+n-1}d|p|\leq\con\,|x|^{-n-\ga}\,.
\end{equation*}
We multiply the remaining integral by $x$   and integrate by parts (Gauss' theorem). This yields
\begin{multline*}
 |x|\bigg|\int\limits_{\ |p|\geq|x|^{-1}}\! F(p)e^{-ip\cdot x}\,d^np\,\bigg|\leq\con\int\limits_{\ |p|=|x|^{-1}}\! |F(p)|\,d|S|(p)\\
 +\con\bigg|\int\limits_{\ |p|\geq|x|^{-1}}\! \p F(p)e^{-ip\cdot x}\,d^np\,\bigg|\,,
\end{multline*}
where $dS(p)$ is the dual integration element on $|p|=|x|^{-1}$. The first term on the rhs is bounded by $\con|x|^{-n-\ga+1}$. The same is true for the second term, if $\ga<-n+1$. If $\ga=-n+1$ we multiply the integral in the second term by $x$ and integrate by parts, which yields two terms analogous to those in the relation above (with one derivative more), both bounded by $\con|x|$. This closes the proof.
\end{proof}

\section{An inequality}\label{inequality}

\begin{lem}\label{inequal}
 Let $f\in L^s(X,d\mu)$, $s\in(0,\infty)$, and $h$ be a~measurable function on $X$. Then for each $\ep\in\<0,1\>$:
\begin{equation}\label{ineq}
 \|h^\ep f\|_s\leq\|f\|_s^{1-\ep}\|hf\|_s^\ep\,.
\end{equation}
In particular, for $f\in L^2(\mR^n,d^nx)$, $\beta=l+\ep$:
\begin{equation}\label{ineqf}
\begin{gathered}
 \||x^i|^\beta f\|_2\leq\|\p_i^l\wh{f}\|_2^{1-\ep}\|\p_i^{l+1}\wh{f}\|_2^\ep\,,\\[1ex]
 \||x|^\beta f\|_2\leq\Big(\sum_{i_1\ldots i_l}\|\p_{i_1}\ldots\p_{i_l}\wh{f}\|_2^2\Big)^\frac{1-\ep}{2}
 \Big(\sum_{j_1\ldots j_{l+1}}\|\p_{j_1}\ldots\p_{j_{l+1}}\wh{f}\|_2^2\Big)^\frac{\ep}{2}\,.
\end{gathered}
\end{equation}
\end{lem}
\begin{proof}
 By the H\"older inequality, if $\nu$ is a~finite, normalized measure on $X$, then for $q\geq1$ there is $\|F\|_{L^1(X,d\nu)}\leq \|F\|_{L^q(X,d\nu)}$. Take
\[
 d\nu(x)=\|f\|_s^{-s}|f(x)|^sd\mu(x)\,,\quad F=h^{s/q}\,.
\]
The inequality then reads:
\begin{equation*}
 \|f\|_s^{-s}\int|h(x)|^{s/q}|f(x)|^sd\mu(x)\leq\|f\|_s^{-s/q}\Big(\int|h(x)f(x)|^sd\mu(x)\Big)^{1/q}\,.
\end{equation*}
Setting here $q=1/\ep$ and taking the $1/s$-power of both sides we arrive at~\eqref{ineq}. The inequalities \eqref{ineqf} are simple applications.
\end{proof}

\section{A lemma on integrals}\label{integrals}

\begin{lem}\label{mumu}
 Let $\mu_i$, $i=1,2$, be Borel measures on $\mR$, such that
\begin{equation*}
 \mu_1((-\infty,a\>)\leq \mu_2((-\infty,a\>)\quad \text{for all}\quad a\in\mR\,.
\end{equation*}
If $f:\mR\mapsto\<0,+\infty\>$ is  non-increasing, then
\begin{equation*}
 \int_{\mR}fd\mu_1\leq \int_{\mR}fd\mu_2\,.
\end{equation*}
\end{lem}
\begin{proof}
Let $\chi(\W)$ denote the characteristic function of the set $\W$. By the Lebesgue theorem  $\dsp\mu((-\infty,a))=\lim_{n\to\infty}\int\chi((-\infty,a-n^{-1}\>)\,d\mu$, so the inequality of measures also extends to $\mu_1((-\infty,a))\leq \mu_2((-\infty,a))$. We~define a~sequence of step functions
\begin{equation*}
 f_N=2^{-N}\sum_{k=1}^{2^{2N}}\chi\Big(f^{-1}\big(\<k2^{-N},+\infty\>\big)\Big)\,.
\end{equation*}
This formula may be rephrased in this way: if $f(x)\in\<k2^{-N},(k+1)2^{-N})$ for $k\in\{0,1,\ldots,2^{2N}-1\}$, then $f_N(x)=k2^{-N}$, and if $f(x)\geq2^N$, then \mbox{$f_N(x)=2^N$}. It follows that $f_N(x)\nearrow f(x)$ for all $x$. Therefore,
\begin{equation*}
 \int_{\mR}f\,d\mu_i=\lim_{N\to\infty}\int f_N\,d\mu_i=\lim_{N\to\infty}2^{-N}\sum_{k=1}^{2^{2N}}\mu_i\big(f^{-1}(\<k2^{-N},+\infty\>)\big)\,.
\end{equation*}
As $f$ is nonincreasing, each of the sets $f^{-1}(\<k2^{-N},+\infty\>)$ is either of the form $(-\infty,a)$ or $(-\infty,a\>$, and the inequalities of measures now give the result.
\end{proof}

\end{document}